\documentclass[11pt, a4paper]{article}

\usepackage[utf8]{inputenc}
\usepackage{amsmath, amssymb, amsthm}
\usepackage{mathtools}
\usepackage{graphicx}
\usepackage{booktabs}
\usepackage{hyperref}
\usepackage{xcolor}
\usepackage{listings}
\usepackage{algorithm}
\usepackage{algpseudocode}
\usepackage{multirow}
\usepackage{enumitem}
\usepackage{geometry}
\usepackage{caption}
\usepackage{subcaption}
\usepackage{tikz}
\usepackage{framed}

\hypersetup{colorlinks=true, linkcolor=black, citecolor=black, urlcolor=blue}

\newtheorem{theorem}{Theorem}[section]
\newtheorem{lemma}[theorem]{Lemma}
\newtheorem{corollary}[theorem]{Corollary}
\newtheorem{proposition}[theorem]{Proposition}
\newtheorem{definition}[theorem]{Definition}
\newtheorem{assumption}[theorem]{Assumption}

\definecolor{listbg}{gray}{0.97}
\lstdefinestyle{mono}{
  basicstyle=\ttfamily\footnotesize,
  backgroundcolor=\color{listbg},
  breaklines=true,
  breakatwhitespace=false,
  columns=fullflexible,
  frame=single,
  framesep=5pt,
  xleftmargin=4pt,
  showstringspaces=false,
}
\title{\textbf{NARAD: Non-colluding Aggregator-oblivious Record-And-Decrypt}\\[3pt]
{\large Blockchain-Verifiable Voting via Packed Paillier Encryption}}
\author{Akshit Vakati Venkata \quad Rajat Dugar \quad Ayush Adarsh \\ \textit{Indian Institute of Technology Madras}}
\date{}

\begin{document}
\maketitle

\begin{abstract}
Electronic voting must keep individual ballots private while letting anyone verify the final tally. This paper presents an architecture that meets both goals without a trusted key dealer: each voter encrypts a ballot in the browser with a self-generated secret key under the Paillier additive homomorphic cryptosystem, and no party ever holds every key. Two server roles divide the tally. A \emph{collector} combines the voters' per-ballot auxiliary values into a single group element; an \emph{aggregator} uses that element to cancel the voters' random masks inside the homomorphic product and recover the exact vote sum, learning the result but no individual ballot. The Solana blockchain records every ciphertext immutably and enforces the election lifecycle, while a native C library (libtommath) performs the heavy modular arithmetic. The paper states six assumptions under which the protocol is correct and prove product homomorphism, mask cancellation, and sum recovery; privacy rests on the assumptions standard for aggregator-oblivious encryption: the Decisional Composite Residuosity (DCR) assumption for the additive layer together with a Diffie--Hellman-style assumption on the masking base. A bit-packing scheme places an entire multi-candidate ballot in one ciphertext, cutting client work, on-chain transactions, storage, and tally cost by a factor of $k$ (the candidate count); the slot width $b$ is free, with only the product $k \cdot b$ bounded by $\log_2 N$. With $b = 25$ and a 255-bit modulus the scheme supports ten candidates and up to $2^{25}-1 = 33\,554\,431$ votes per candidate, about 335 million ballots in total, and the proof-of-concept tallies 50{,}000 ballots exactly in under one second on its real parameters. Finally, we show how running the collector and aggregator inside attested secure enclaves makes the tally tamper-resistant and prevents a single host from colluding across the two roles to deanonymize voters. The proof-of-concept implementation is open-source.\footnote{\url{https://github.com/Akshit11318/narad}} A worked numerical example in the appendix reproduces the full pipeline on those same parameters.
\end{abstract}


\section{Introduction}\label{sec:intro}

Electronic voting must satisfy two requirements that usually pull against each other: \emph{ballot privacy} (no party learns how an individual voted) and \emph{tally verifiability} (the published result is publicly checkable). Mix-net systems give privacy but need trusted authorities to shuffle and decrypt. Additive homomorphic encryption gives a cleaner path: the tally is computed directly on ciphertexts, so no individual ballot is ever decrypted.

\subsection{The aggregator-oblivious model}

In the aggregator-oblivious model~\cite{lem14b}, voters encrypt their data with self-generated secret keys and no coordination with a key dealer. No single party holds every secret key. The aggregator computes the sum of the encrypted values from auxiliary information supplied by a collector, but cannot decrypt any single ciphertext.

The mechanism rests on one structural choice: each voter's secret key $sk_i$ plays two roles at once. It is the encryption randomness $r_i$ of the masking term, and it is the exponent of the auxiliary value $aux_i = pk_A^{sk_i}$. When the aggregator raises the ciphertext product to its own key $sk_A$ and divides by the collected auxiliary product, the random masks cancel exactly and only the encrypted sum remains.

\subsection{Contributions}

\begin{enumerate}[leftmargin=*]
    \item A system model and architecture that combine aggregator-oblivious Paillier encryption with blockchain immutability: clients encrypt in the browser through WebAssembly, and aggregation runs in a native C library (Sections~\ref{sec:model} and~\ref{sec:implementation}).
    \item Six explicit assumptions and complete correctness proofs for the aggregation pipeline, covering product homomorphism, mask cancellation, and sum recovery (Sections~\ref{sec:protocol}--\ref{sec:proofs}).
    \item A privacy analysis under the DCR assumption, showing that the aggregator, collector, and network observer each learn nothing about individual ballots (Section~\ref{sec:security}).
    \item A bit-packing scheme that fits a multi-candidate ballot into one ciphertext, with a tunable slot width $b$, capacity bounds, and a quantified account of why packing matters (Section~\ref{sec:packing}).
    \item A proof-of-concept measured on its own real parameters, with multi-scale results and a worked numerical example computed on those parameters (Sections~\ref{sec:implementation}--\ref{sec:evaluation}, Appendix~\ref{sec:example}).
    \item A discussion of how the design scales to more or fewer candidates and voters, how on-chain cost trades against capacity, how preferential voting fits the packing model, and how secure enclaves can make the tally tamper-resistant (Section~\ref{sec:future}).
\end{enumerate}

\subsection{Outline}

Section~\ref{sec:prelim} recalls the Paillier primitives and fixes notation. Section~\ref{sec:related} positions our work against prior systems and argues why the aggregator-oblivious model is preferable to a trusted key dealer in an enclave. Section~\ref{sec:guarantees} states the guarantees and their assumptions upfront. Section~\ref{sec:model} describes the parties and trust model. Section~\ref{sec:protocol} states the protocol and its assumptions; Section~\ref{sec:proofs} proves correctness; Section~\ref{sec:security} analyzes security; Section~\ref{sec:packing} covers vote packing; Section~\ref{sec:algorithms} gives the algorithms; Section~\ref{sec:implementation} describes the implementation and its architecture on real parameters; and Section~\ref{sec:evaluation} reports measurements. Section~\ref{sec:future} discusses scaling, threshold aggregation, and future work, and Section~\ref{sec:conclusion} concludes. Appendix~\ref{sec:example} works a full example end to end.

\section{Preliminaries}\label{sec:prelim}

\subsection{Paillier encoding and the \texorpdfstring{$L$}{L}-function}

We use the additive structure of the Paillier cryptosystem~\cite{paillier1999public}. For a modulus $N = pq$ with $p, q$ distinct odd primes, and the base $1 + N$, the binomial theorem gives, for any non-negative integer $m$,
\begin{equation}
(1 + N)^m = \sum_{t \geq 0} \binom{m}{t} N^t \equiv 1 + mN \pmod{N^2}, \label{eq:binom}
\end{equation}
since every term with $t \geq 2$ is divisible by $N^2$. Encoding is additively homomorphic: $(1+N)^{m_1}(1+N)^{m_2} \equiv (1+N)^{m_1+m_2} \pmod{N^2}$, so multiplying encodings adds plaintexts modulo $N$.

To hide a plaintext, the encoding is multiplied by a random element of $\mathbb{Z}_{N^2}^*$. In this work that mask is $H^{r}$ for a public base $H \in \mathbb{Z}_{N^2}^*$ and a secret exponent $r$, so a ciphertext is $c = H^{r}\,(1+N)^m \bmod N^2$.

The plaintext is recovered with the $L$-function
\begin{equation}
L(u) = \frac{u - 1}{N}, \qquad u \equiv 1 \pmod N,
\end{equation}
where the division is exact over the integers. For $m < N$, Equation~\ref{eq:binom} gives $(1+N)^m \bmod N^2 = 1 + mN$, hence $L(1 + mN) = m$. The protocol of Section~\ref{sec:protocol} removes the mask before applying $L$.

\subsection{Hardness assumptions}

\begin{definition}[Decisional Composite Residuosity]\label{def:dcr}
Given $N = pq$ and $z \in \mathbb{Z}_{N^2}^*$, it is computationally infeasible to decide whether $z$ is an $N$-th residue modulo $N^2$ (that is, whether $\exists\, y \in \mathbb{Z}_{N^2}^*$ with $z \equiv y^N \pmod{N^2}$) or a uniformly random element of $\mathbb{Z}_{N^2}^*$.
\end{definition}

Privacy of the construction relies on two assumptions, both standard for aggregator-oblivious encryption~\cite{shi2011, joye2013, lem14b}. First, the additive Paillier layer $(1+N)^{m}$ is semantically secure under DCR (Definition~\ref{def:dcr})~\cite{paillier1999public}. Second, the masking term $H^{sk_i}$, for a secret exponent $sk_i$ and the public base $H$, is pseudorandom in the subgroup $\langle H\rangle$ it generates, a Diffie--Hellman-style assumption on $H$. Note that the mask lies in $\langle H\rangle$, not in all of $\mathbb{Z}_{N^2}^*$, so we do \emph{not} claim a ciphertext is indistinguishable from a uniform group element. We do not reprove the scheme; its privacy under these two assumptions is established in~\cite{shi2011, joye2013, lem14b}, and Section~\ref{sec:security} states the guarantees our deployment relies on.

\subsection{Notation}

Table~\ref{tab:notation} collects the symbols used throughout.

\begin{table}[h]
\centering
\caption{Notation.}
\label{tab:notation}
\begin{tabular}{ll}
\toprule
Symbol & Meaning \\
\midrule
$N = pq$ & Paillier modulus; $p, q$ distinct odd primes \\
$N^2$ & ciphertext modulus, defining the group $\mathbb{Z}_{N^2}^*$ \\
$H$ & public base element of $\mathbb{Z}_{N^2}^*$ \\
$sk_A,\ pk_A$ & aggregator secret/public key, $pk_A = H^{sk_A} \bmod N^2$ \\
$sk_i$ & voter $i$'s secret key, reused as the masking randomness $r_i$ \\
$x_i$ & voter $i$'s packed ballot (plaintext), $x_i < N$ \\
$c_i$ & voter $i$'s ciphertext $H^{sk_i}(1+N)^{x_i} \bmod N^2$ \\
$aux_i$ & voter $i$'s auxiliary value $pk_A^{sk_i} \bmod N^2$ \\
$L(\cdot)$ & $L$-function $L(u) = (u-1)/N$ \\
$n$ & number of voters \\
$k$ & number of candidates \\
$b$ & bits per packing slot (implementation uses $b = 25$) \\
$v_{i,j}$ & voter $i$'s entry for candidate $j$; $v_{i,j} \in \{0,1\}$ \\
$S$ & recovered sum $\sum_{i} x_i$ \\
$\Sigma$ & sum of secret keys $\sum_{i} sk_i$ \\
$\Pi$ & ciphertext product $\prod_i c_i \bmod N^2$ \\
$P,\ P'$ & tally intermediates $\Pi^{sk_A}$ and $P\cdot aux^{-1} \bmod N^2$ \\
\bottomrule
\end{tabular}
\end{table}

\section{Related Work and Positioning}\label{sec:related}

We position our work against the main families of secure voting and aggregation, and argue why the aggregator-oblivious model is preferable to simpler alternatives.

\paragraph{Helios~\cite{adida2008helios}} is a web-based open-audit voting system using homomorphic encryption. It relies on a \emph{single trusted authority} to hold the decryption key and compute the tally. If that authority is compromised, all ballot privacy is lost. Our design eliminates this trust bottleneck: no single party holds every secret key, and the aggregator learns only the sum, not individual ballots. Helios also does not use a blockchain, so its audit trail depends on the honesty of the server that stores ciphertexts; we record every ciphertext on Solana, giving an immutable record that no party can alter.

\paragraph{Civitas~\cite{kiayias2008civitas}} extends Helios with coercion resistance using distributed plaintext-equivalence tests and enforced eligibility. It requires a set of \emph{supervisor} authorities that jointly manage decryption, a threshold scheme, but one that still concentrates trust in a small committee. Our aggregator-oblivious model goes further: \emph{each voter} is their own key dealer, so there is no committee to corrupt. The price we pay is the absence of coercion resistance, which Civitas handles through receipt-freeness and forced-abstention protections. We view coercion resistance as orthogonal and complementary.

\paragraph{Aggregator-oblivious encryption.} Shi et al.~\cite{shi2011} introduced privacy-preserving aggregation in which users encrypt with self-generated keys; Joye and Libert~\cite{joye2013} gave a scalable Paillier-based scheme, and Leontiadis et al.~\cite{lem14b} relaxed its trust by removing the trusted dealer and periodic key updates. These works are theoretical: they define protocols and prove security but do not address browser-side encryption, multi-candidate packing, blockchain audit trails, or native-code performance. Our contribution is to take this line of protocols from theory to a measured system, with bit-packing for multi-candidate ballots, WebAssembly for client-side encryption, a native C library for server-side aggregation, and a Solana program for immutable recording.

\paragraph{Secure aggregation in federated learning.} Mansouri et al.~\cite{popets2023} systematize cryptographic secure-aggregation schemes for federated learning. Turbo-Aggregate~\cite{turbo2021} reaches $O(N \log N)$ aggregation via secret sharing and coding, tolerating dropouts, a property voting does not need (a dropped ballot is not counted) but that complicates the protocol. These systems target model gradients, not binary ballots, and do not use blockchain for input integrity. Our protocol is simpler (no secret sharing, no coding), targets a different application (voting), and adds on-chain immutability.

\paragraph{Why not run a trusted key dealer inside a secure enclave?}
A natural alternative to aggregator-oblivious encryption is to keep a conventional trusted-key-dealer Paillier scheme but run the dealer (and the tallying authority) inside an attested secure enclave such as Intel SGX or AWS Nitro Enclaves. The enclave would hold the factorization of $N$ or the decryption key, compute the tally, and attest to the result. This is simpler: one party, one key, and removes the need for the collector/aggregator split entirely.

We do not pursue this route for three reasons. \emph{First}, it concentrates all cryptographic trust in the enclave's correctness and the hardware manufacturer's honesty. SGX has suffered practical side-channel attacks~\cite{sgaxe2020, lvi2020} that can extract keys from ``isolated'' enclaves; a single successful attack on the dealer enclave compromises every ballot in the election, not only one tally. In our design, by contrast, a compromised aggregator learns only the sum. Individual ballots remain protected by DCR even if $sk_A$ leaks. \emph{Second}, the aggregator-oblivious model gives a cryptographic privacy guarantee that does not depend on hardware at all: privacy rests on DCR, a number-theoretic assumption, not on the physical integrity of a chip. Enclaves can \emph{strengthen} integrity (Section~\ref{sec:enclave}) but they are not the \emph{basis} of privacy. \emph{Third}, the key-dealer model requires the dealer to be online during tallying, creating an availability bottleneck; our aggregator needs only the collected $aux$, which is a single group element that can be computed and forwarded by anyone who receives the $aux_i$ values.

In short, enclaves are a good \emph{addition} to our design (making the tally tamper-resistant) but a poor \emph{replacement} for the cryptographic separation of duties that aggregator-oblivious encryption provides.

\paragraph{Zero-knowledge proofs for ballot validity.}
A limitation of our current design is the lack of cryptographic enforcement of Assumption~\ref{ass:binary} ($v_{i,j} \in \{0,1\}$). A malicious voter could submit a ciphertext encoding an out-of-range value. The standard remedy is a zero-knowledge range proof that the committed value lies in $\{0,1\}$, proven alongside the ciphertext without revealing the vote. Pedersen commitments~\cite{pedersen1991} provide the hiding commitment $C = g^v h^r \bmod p$; a Schnorr-style proof of knowledge~\cite{schnorr1991} then demonstrates that the prover knows $v$ and $r$ consistent with $C$; and a range proof based on the constraint $v(v-1) = 0$ (equivalently, $v \in \{0,1\}$) can be constructed via Bulletproofs~\cite{bunz2018bulletproofs} or the simpler $\Sigma$-protocol of Camenisch and Stadler~\cite{camenisch1997}. A sum proof (that exactly one candidate is selected, $\sum_j v_{i,j} = 1$) can be built from a Schnorr proof on the aggregate commitment. Our codebase includes WASM-backed implementations of these primitives (modular exponentiation, secure hashing, Fiat--Shamir challenges) but has not yet integrated them into the voting flow; this is the most important piece of future work.

\section{Guarantees and Assumptions Summary}\label{sec:guarantees}

Before diving into the protocol, we state upfront what the system guarantees and under what conditions. Each guarantee names the assumption or proof that backs it.

\subsection{Guarantees provided}

\begin{itemize}[leftmargin=*]
    \item \textbf{Ballot privacy (individual).} No party (aggregator, collector, or network observer) learns how an individual voter voted. \emph{Reason:} Each ciphertext $c_i = H^{sk_i}(1+N)^{x_i} \bmod N^2$ is semantically secure under DCR (Definition~\ref{def:dcr}) and the masking assumption on $H$ (Section~\ref{sec:prelim}). The secret key $sk_i$ never leaves the voter's browser. (Proposition~\ref{prop:privacy}.)
    
    \item \textbf{Aggregator obliviousness.} The aggregator learns only the tally $S = \sum x_i$, not any individual $x_i$. \emph{Reason:} The aggregator sees $\{c_i\}$ and the product $aux = \prod aux_i$ but never an individual $aux_i$. Recovering a single $x_i$ from $c_i$ would break DCR. (Proposition~\ref{prop:agg}.)
    
    \item \textbf{Collector obliviousness.} The collector learns nothing about any vote. \emph{Reason:} Each $aux_i = pk_A^{sk_i}$ depends only on $(pk_A, sk_i)$ and is independent of the ballot $x_i$. The collector never sees a ciphertext. (Proposition~\ref{prop:collector}.)
    
    \item \textbf{Tally correctness.} The recovered sum $R$ equals the true sum $S = \sum x_i$ exactly. \emph{Reason:} Proven under Assumptions~\ref{ass:modulus}--\ref{ass:sum} via product homomorphism (Theorem~\ref{thm:product}), mask cancellation (Theorem~\ref{thm:mask}), and sum recovery (Theorem~\ref{thm:sum}).
    
    \item \textbf{Per-candidate count correctness.} The extracted count for each candidate equals the true number of votes for that candidate. \emph{Reason:} The 25-bit packing preserves slots independently under homomorphic addition, and extraction isolates each slot without cross-talk. (Corollary~\ref{cor:votes}, Assumptions~\ref{ass:packing} and~\ref{ass:binary}.)
    
    \item \textbf{Input integrity (tamper-evident).} The set of ciphertexts used in aggregation can be checked against the immutable on-chain record. \emph{Reason:} Each $c_i$ is a finalized Solana transaction; any insertion, deletion, or modification diverges from the chain and is detectable by any auditor.
    
    \item \textbf{Output integrity (deterministic).} Anyone holding the on-chain ciphertexts and $aux$ can recompute $R$ and verify the published result. \emph{Reason:} $R$ is a deterministic function of public inputs; a dishonest aggregator cannot publish a different valid tally.
    
    \item \textbf{No key coordination.} Voters generate their own secret keys without a dealer. \emph{Reason:} The protocol's mask cancellation works because $sk_i$ serves as both the encryption randomness and the auxiliary exponent (Assumption~\ref{ass:keyrand}), so no external key distribution is needed.
\end{itemize}

\subsection{What is \emph{not} guaranteed}

\begin{itemize}[leftmargin=*]
    \item \textbf{No range enforcement.} The system does not cryptographically verify $v_{i,j} \in \{0,1\}$ (Assumption~\ref{ass:binary} is trusted, not enforced). ZK range proofs would close this gap.
    \item \textbf{No coercion resistance.} A voter can reveal $sk_i$ to prove how they voted.
    \item \textbf{No collusion resistance (without enclaves).} Privacy assumes the collector and aggregator do not pool their views. Secure enclaves (Section~\ref{sec:enclave}) can make this a hardware-enforced guarantee.
    \item \textbf{No availability guarantee.} A withholding aggregator can refuse to publish. Threshold decryption across multiple aggregators (Section~\ref{sec:threshold}) would address this.
\end{itemize}

\section{System Model}\label{sec:model}

The protocol involves four kinds of party. This section describes them at the level of roles, knowledge, and trust; their concrete realization as software components is deferred to Section~\ref{sec:implementation}. Figure~\ref{fig:model} shows the parties and the messages between them.

\begin{figure}[h]
\centering
\includegraphics[width=\textwidth]{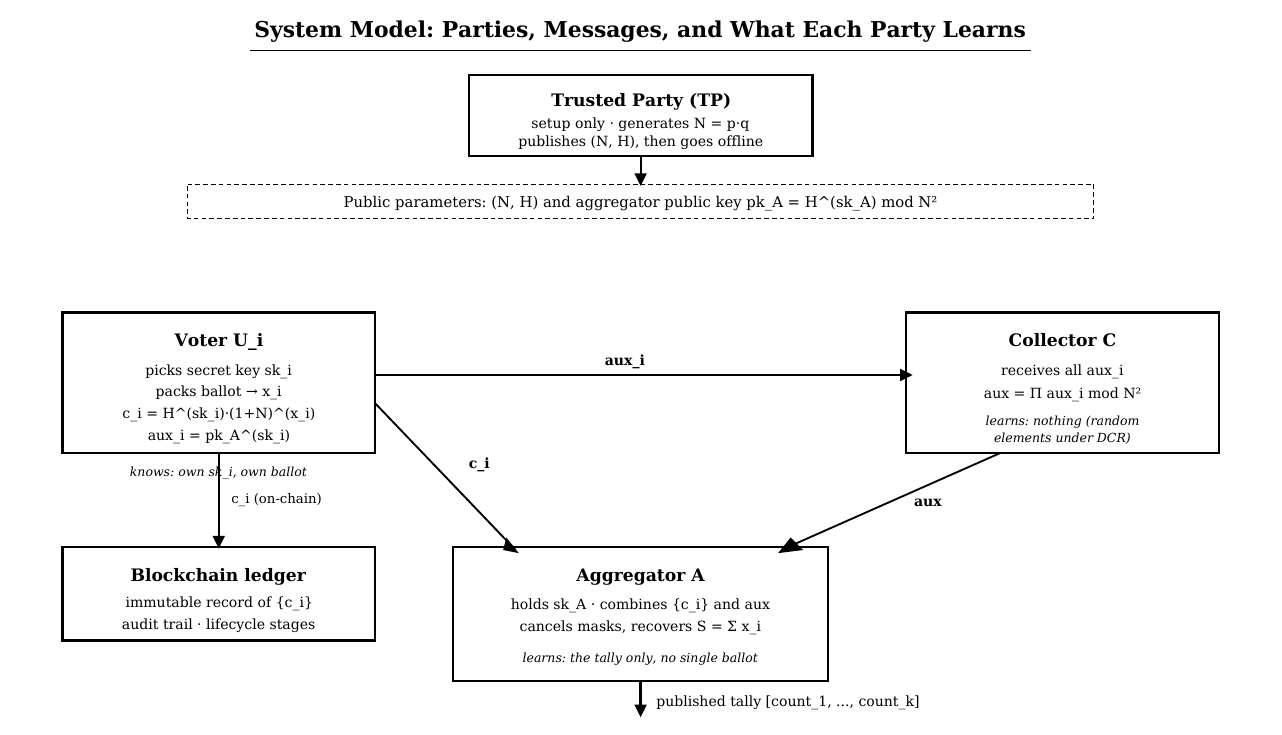}
\caption{System model. The trusted party publishes $(N, H)$ at setup and leaves. Each voter sends a ciphertext $c_i$ to the aggregator (and to the chain) and an auxiliary value $aux_i$ to the collector. The collector forwards the product $aux$; the aggregator recovers and publishes the tally. No party except the voter sees an individual plaintext.}
\label{fig:model}
\end{figure}

\subsection{Parties and roles}

\begin{itemize}[leftmargin=*]
    \item \textbf{Trusted third party ($TP$).} Generates the modulus $N = pq$, picks the base $H$, publishes $(N, H)$, and goes offline. It takes part only in setup and never sees a ballot.
    \item \textbf{Voters ($U_1, \ldots, U_n$).} Each voter picks a secret key $sk_i$, encodes a ballot as a packed integer $x_i$, and produces a ciphertext $c_i$ and an auxiliary value $aux_i$. A voter knows only its own key and ballot.
    \item \textbf{Collector ($C$).} The collector assembles the material that lets the masks be canceled in bulk. Each voter sends it an auxiliary value $aux_i = pk_A^{sk_i} = H^{sk_A \cdot sk_i}$, which ties the voter's secret key $sk_i$ to the aggregator's public key but carries no information about the ballot. The collector multiplies these into one element $aux = \prod_i aux_i = H^{sk_A \sum_i sk_i}$ and forwards it to the aggregator. This single value is exactly what the aggregator needs to strip every per-voter mask $H^{sk_i}$ at once; without it the sum cannot be recovered. The collector never receives a ciphertext $c_i$, and each $aux_i$ depends only on $(pk_A, sk_i)$, so its role is integrity-critical: it must include every $aux_i$ exactly once, but it learns nothing about any vote.
    \item \textbf{Aggregator ($A$).} The aggregator produces the result. At setup it generates its own secret key $sk_A$ and publishes $pk_A = H^{sk_A}$. At tally time it multiplies all ciphertexts into $\Pi = \prod_i c_i$, raises the product to $sk_A$, divides out the collector's $aux$ to cancel the masks, applies the $L$-function, and multiplies by $sk_A^{-1}$ to recover the sum $S = \sum_i x_i$, which it unpacks into per-candidate counts. It holds $sk_A$ and sees every ciphertext and the aggregate $aux$, but never an individual $aux_i$; under the assumptions of Section~\ref{sec:prelim} it learns the tally and nothing about any single ballot. It can refuse to publish but cannot forge a different result, since the tally is fixed by the public ciphertexts and $aux$.
\end{itemize}

\subsection{Trust model}

\begin{itemize}[leftmargin=*]
    \item \textbf{$TP$:} trusted only for setup. It knows the factorization of $N$ but no $sk_i$, and it is offline once $(N, H)$ are published.
    \item \textbf{$A$:} holds $sk_A$. It can compute the sum but not individual votes (under DCR, Definition~\ref{def:dcr}). It can withhold a result but cannot forge one: the sum is fixed by $\{c_i\}$ and $aux$.
    \item \textbf{$C$:} sees only $\{aux_i\}$, which are independent of the ballots and, under DCR, indistinguishable from random group elements. It learns nothing about votes.
    \item \textbf{$U_i$:} knows only its own $sk_i$ and ballot; the plaintext never leaves the voter.
    \item \textbf{Blockchain validators:} untrusted, under Byzantine fault-tolerant consensus.
\end{itemize}

\subsection{High-level flow}

Setup fixes $(N, H)$ and the aggregator key. Each voter independently encrypts and emits $(c_i, aux_i)$; the ciphertexts are also recorded on-chain for an audit trail. The collector aggregates the auxiliary values into a single element, and the aggregator uses it to cancel the per-voter masks and recover the tally in one homomorphic computation. The cryptographic detail of each step follows in Section~\ref{sec:protocol}.

\section{Cryptographic Protocol}\label{sec:protocol}

Recall the four parties of Section~\ref{sec:model}: $TP$, voters $U_1, \ldots, U_n$, collector $C$, and aggregator $A$. Figure~\ref{fig:protocol} traces the full execution from a voter's ballot to the published tally.

\begin{figure}[p]
\centering
\includegraphics[width=\textwidth]{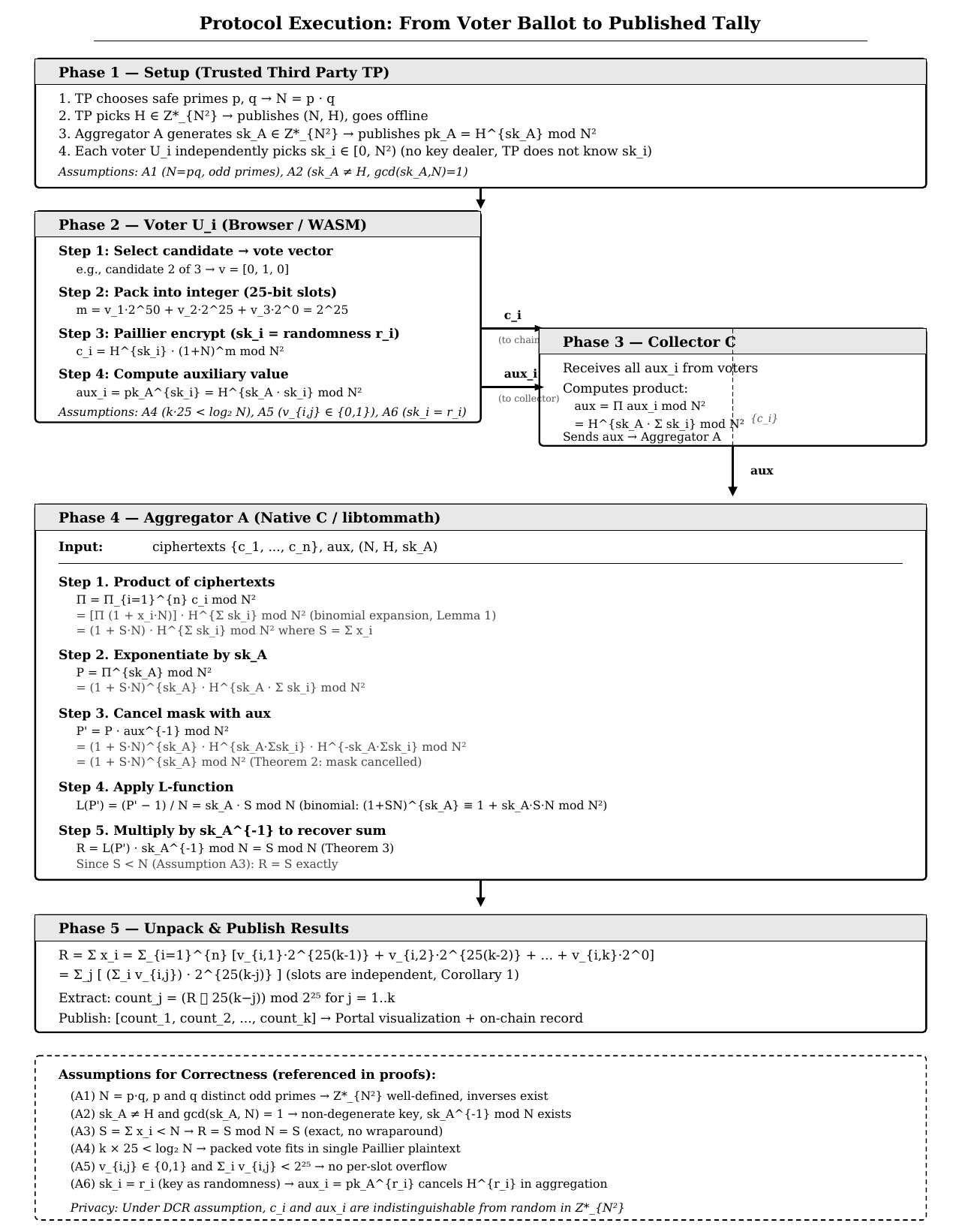}
\caption{Protocol execution. Phase 1: setup by $TP$. Phase 2: a voter encrypts a ballot and computes an auxiliary value. Phase 3: the collector multiplies the auxiliary values. Phase 4: the aggregator runs the Paillier tally. Phase 5: unpack and publish. The dashed box lists the six assumptions used in the proofs.}
\label{fig:protocol}
\end{figure}

\subsection{Assumptions}

The six assumptions below are sufficient for correctness. Each is cited by the proofs of Section~\ref{sec:proofs}.

\begin{assumption}[Valid modulus]\label{ass:modulus}
$N = p \cdot q$ with $p, q$ distinct odd primes. This makes $\mathbb{Z}_{N^2}^*$ a well-defined multiplicative group, gives $\gcd(H, N^2) = 1$ for $H \in \mathbb{Z}_{N^2}^*$, and guarantees that modular inverses exist for every element coprime to $N^2$.
\end{assumption}

\begin{assumption}[Invertible, non-trivial aggregator key]\label{ass:key}
$\gcd(sk_A, N) = 1$, and the resulting public key $pk_A = H^{sk_A} \bmod N^2$ is non-trivial, i.e.\ $pk_A \notin \{1, H\}$. The first condition is the one correctness needs: it makes $sk_A^{-1} \bmod N$ exist for the final recovery step. The second is a directly checkable sanity condition (both $pk_A$ and the comparison are computable without the factorization of $N$) ruling out a degenerate public key. The reference implementation enforces both, having originally failed by setting $sk_A$ equal to $H$.
\end{assumption}

\begin{assumption}[Packing capacity]\label{ass:packing}
$k \cdot b \le \lfloor \log_2 N \rfloor$, where $k$ is the candidate count and $b$ is the bits per slot. Equivalently $2^{k b} \le N$. This keeps every packed vote $x_i$ inside $\mathbb{Z}_N$, i.e.\ $x_i < N$. Only the product $k \cdot b$ is constrained; $b$ is otherwise free. For the 255-bit $N$ used here\footnote{A 255-bit integer has $\lfloor \log_2 N \rfloor = 254$, so the bound on $k \cdot b$ is 254, not 255.}, $\lfloor\log_2 N\rfloor = 254$, so with $b = 25$, $k \leq \lfloor 254/25 \rfloor = 10$.
\end{assumption}

\begin{assumption}[Per-slot bound]\label{ass:binary}
Each $v_{i,j} \in \{0, 1\}$ and, for every candidate $j$, the tally $T_j = \sum_{i=1}^n v_{i,j}$ satisfies $T_j < 2^b$. This keeps each slot from overflowing during homomorphic addition.
\end{assumption}

\begin{assumption}[Sum bound]\label{ass:sum}
$S = \sum_{i=1}^{n} x_i < N$. This makes the final modular reduction exact, $R = S \bmod N = S$, with no wraparound. It is implied by Assumptions~\ref{ass:packing} and~\ref{ass:binary}: if every slot stays below $2^b$, then $S \le \sum_{j=1}^{k}(2^b-1)\,2^{b(k-j)} = 2^{kb} - 1 < N$.
\end{assumption}

\begin{assumption}[Key as randomness]\label{ass:keyrand}
The voter's secret key is reused as the masking randomness, $r_i = sk_i$. This is the structural assumption behind mask cancellation. Because $aux_i = pk_A^{sk_i} = H^{sk_A \cdot sk_i}$, the factor $H^{sk_i}$ in the ciphertext product, once raised to $sk_A$, becomes $H^{sk_A \cdot sk_i}$, which $aux_i$ cancels exactly.
\end{assumption}

\subsection{Setup phase}

$TP$ selects primes $p, q$, computes $N = pq$, and picks $H \in \mathbb{Z}_{N^2}^*$. It publishes $(N, H)$ and goes offline. The aggregator $A$ generates $sk_A$ satisfying Assumption~\ref{ass:key} and publishes $pk_A = H^{sk_A} \bmod N^2$. Each voter $U_i$ independently samples a secret key $sk_i \in_R [2, N^2)$.

\subsection{Encryption phase}

Voter $U_i$ encodes the ballot as a packed integer $x_i$ (Section~\ref{sec:packing}; in the encryption formula below, $x_i$ is written as $m$ for clarity) and computes
\begin{align}
c_i &= H^{sk_i} \cdot (1 + N)^{x_i} \bmod N^2, \label{eq:enc}\\
aux_i &= pk_A^{sk_i} = H^{sk_A \cdot sk_i} \bmod N^2. \label{eq:auxenc}
\end{align}
The ciphertext $c_i$ goes to the aggregator and is recorded on-chain; $aux_i$ goes to the collector. Note that $aux_i$ depends only on $(pk_A, sk_i)$ and not on the ballot $x_i$.

\subsection{Collection phase}

The collector $C$ computes the product of the auxiliary values:
\begin{equation}
aux = \prod_{i=1}^{n} aux_i = H^{sk_A \cdot \sum_{i=1}^n sk_i} \bmod N^2, \label{eq:collect2}
\end{equation}
and sends $aux$ to the aggregator.

\subsection{Aggregation phase}

The aggregator recovers the sum in five steps:
\begin{align}
\Pi &= \prod_{i=1}^{n} c_i \bmod N^2, \label{eq:step1}\\
P &= \Pi^{sk_A} \bmod N^2, \label{eq:step2}\\
P' &= P \cdot aux^{-1} \bmod N^2, \label{eq:step3}\\
L(P') &= \frac{P' - 1}{N}, \label{eq:step4}\\
R &= L(P') \cdot sk_A^{-1} \bmod N. \label{eq:step5}
\end{align}
Under Assumptions~\ref{ass:modulus}--\ref{ass:sum}, $R = \sum_{i=1}^n x_i$ exactly (Theorem~\ref{thm:sum}).

\section{Correctness Proofs}\label{sec:proofs}

We prove that the pipeline recovers $S = \sum x_i$ through three theorems and a supporting lemma. Each proof names the assumptions it uses.

\begin{lemma}[Binomial expansion in $\mathbb{Z}_{N^2}$]\label{lem:binom}
For non-negative integers $a_1, \ldots, a_n$,
\[\prod_{i=1}^{n} (1 + a_i N) \equiv 1 + \left(\sum_{i=1}^{n} a_i\right) N \pmod{N^2}.\]
\end{lemma}

\begin{proof}
By induction on $n$.

\textbf{Base case} ($n = 1$): $1 + a_1 N \equiv 1 + a_1 N \pmod{N^2}$, which holds trivially.

\textbf{Inductive step}: assume the claim for $n - 1$ terms. Then
\begin{align*}
\prod_{i=1}^{n}(1 + a_i N) &= \underbrace{\prod_{i=1}^{n-1}(1 + a_i N)}_{\equiv\, 1 + (\sum_{i=1}^{n-1} a_i) N \pmod{N^2}} \cdot\, (1 + a_n N) \\
&\equiv \left(1 + \textstyle\sum_{i=1}^{n-1} a_i \cdot N\right)(1 + a_n N) \pmod{N^2} \\
&= 1 + \textstyle\sum_{i=1}^{n-1} a_i \cdot N + a_n N + \left(\textstyle\sum_{i=1}^{n-1} a_i\right) a_n N^2 \\
&\equiv 1 + \left(\textstyle\sum_{i=1}^{n} a_i\right) N \pmod{N^2},
\end{align*}
since the $N^2$ term vanishes modulo $N^2$.
\end{proof}

\begin{theorem}[Ciphertext product homomorphism]\label{thm:product}
Under Assumption~\ref{ass:modulus},
\[\prod_{i=1}^{n} c_i \equiv \left(1 + S \cdot N\right) \cdot H^{\Sigma} \pmod{N^2},\]
where $S = \sum_{i=1}^n x_i$ and $\Sigma = \sum_{i=1}^n sk_i$. (The congruence holds for any non-negative $x_i$; the packing and sum bounds enter only at recovery, Theorem~\ref{thm:sum} and Corollary~\ref{cor:votes}.)
\end{theorem}

\begin{proof}
From Equation~\ref{eq:enc}, $c_i = H^{sk_i} \cdot (1+N)^{x_i} \bmod N^2$. By Equation~\ref{eq:binom}, $(1+N)^{x_i} \equiv 1 + x_i N \pmod{N^2}$, so $c_i \equiv H^{sk_i} \cdot (1 + x_i N) \pmod{N^2}$. Taking the product,
\[\prod_{i=1}^n c_i \equiv \left(\prod_{i=1}^n H^{sk_i}\right) \cdot \left(\prod_{i=1}^n (1 + x_i N)\right) \pmod{N^2}.\]
The first factor is $H^{\sum sk_i} = H^{\Sigma}$. For the second, Lemma~\ref{lem:binom} gives $\prod_{i=1}^n (1 + x_i N) \equiv 1 + S \cdot N \pmod{N^2}$. Combining the two factors yields the claim.
\end{proof}

\begin{theorem}[Mask cancellation]\label{thm:mask}
Under Assumptions~\ref{ass:modulus},~\ref{ass:key}, and~\ref{ass:keyrand},
\[P' = P \cdot aux^{-1} \equiv (1 + S \cdot N)^{sk_A} \pmod{N^2}.\]
\end{theorem}

\begin{proof}
From Theorem~\ref{thm:product}, $\Pi \equiv (1 + S \cdot N) \cdot H^{\Sigma} \pmod{N^2}$. Raising to $sk_A$,
\[P = \Pi^{sk_A} \equiv (1 + S \cdot N)^{sk_A} \cdot H^{sk_A \cdot \Sigma} \pmod{N^2}.\]
By Equation~\ref{eq:collect2} and Assumption~\ref{ass:keyrand}, $aux = H^{sk_A \cdot \Sigma} \bmod N^2$. By Assumption~\ref{ass:modulus}, $\gcd(H, N^2) = 1$, so $aux \in \mathbb{Z}_{N^2}^*$ and $aux^{-1} \bmod N^2$ exists. Therefore
\begin{align*}
P' = P \cdot aux^{-1} &\equiv (1 + S \cdot N)^{sk_A} \cdot H^{sk_A \Sigma} \cdot H^{-sk_A \Sigma} \\
&\equiv (1 + S \cdot N)^{sk_A} \pmod{N^2}.
\end{align*}
The mask $H^{sk_A \Sigma}$ cancels exactly.
\end{proof}

\begin{theorem}[Sum recovery]\label{thm:sum}
Under Assumptions~\ref{ass:modulus},~\ref{ass:key}, and~\ref{ass:sum},
\[R = L(P') \cdot sk_A^{-1} \bmod N = S = \sum_{i=1}^n x_i.\]
\end{theorem}

\begin{proof}
By the binomial theorem (Equation~\ref{eq:binom} with exponent $sk_A$ and base increment $S \cdot N$),
\[(1 + S N)^{sk_A} \equiv 1 + sk_A \cdot S \cdot N \pmod{N^2},\]
since every term of order $\geq 2$ in $SN$ carries $N^2$. Reducing the canonical representative of $P'$ in $[0, N^2)$, the coefficient of $N$ lies in $[0, N)$, so
\[P' = 1 + \big( (sk_A \cdot S) \bmod N \big)\, N, \qquad L(P') = \frac{P' - 1}{N} = (sk_A \cdot S) \bmod N.\]
By Assumption~\ref{ass:key}, $\gcd(sk_A, N) = 1$, so $sk_A^{-1} \bmod N$ exists and
\[R = L(P') \cdot sk_A^{-1} \bmod N = (sk_A \cdot S \cdot sk_A^{-1}) \bmod N = S \bmod N.\]
By Assumption~\ref{ass:sum}, $S < N$, so $S \bmod N = S$.
\end{proof}

\begin{corollary}[Vote-count correctness]\label{cor:votes}
Under Assumptions~\ref{ass:packing} and~\ref{ass:binary}, the extracted count for each candidate equals the true tally:
\[\mathrm{count}_j = \sum_{i=1}^{n} v_{i,j} \quad \text{for all } j \in \{1, \ldots, k\}.\]
\end{corollary}

\begin{proof}
By Theorem~\ref{thm:sum}, $R = \sum_{i=1}^n x_i$. Each $x_i = \sum_{j=1}^k v_{i,j} \cdot 2^{b(k-j)}$, so
\[R = \sum_{i=1}^n \sum_{j=1}^k v_{i,j} \cdot 2^{b(k-j)} = \sum_{j=1}^k \underbrace{\left(\sum_{i=1}^n v_{i,j}\right)}_{T_j} \cdot 2^{b(k-j)}.\]
Because each $T_j < 2^b$ (Assumption~\ref{ass:binary}), the slots do not interfere. The extraction $\mathrm{count}_j = \lfloor R / 2^{b(k-j)} \rfloor \bmod 2^b$ then isolates $T_j$, and since $v_{i,j} \in \{0,1\}$, $T_j$ is the true count.
\end{proof}

\section{Security Analysis}\label{sec:security}

\subsection{Cryptographic foundation}

Privacy rests on the two assumptions of Section~\ref{sec:prelim}: DCR for the additive Paillier layer, and pseudorandomness of the masking term $H^{sk_i}$ in $\langle H\rangle$ under a Diffie--Hellman-style assumption on $H$. Under these assumptions the masked encoding $c_i = H^{sk_i}(1+N)^{x_i} \bmod N^2$ is a semantically secure encryption of $x_i$; the aggregator-oblivious construction is proven private under exactly these assumptions in~\cite{shi2011, joye2013, lem14b}. We do not reprove that result. The propositions below state the guarantees in our setting, where the collector and aggregator are distinct, non-colluding parties (collusion is discussed under Limitations). Importantly, no single honest-but-curious party sees both $c_i$ and the matching $aux_i$: the collector sees only $\{aux_i\}$, while the aggregator sees $\{c_i\}$ and the \emph{product} $aux$, never an individual $aux_i$.

\subsection{Privacy guarantees}

\begin{proposition}[Individual vote privacy]\label{prop:privacy}
Under the assumptions of Section~\ref{sec:prelim}, no efficient adversary that sees $c_i$, together with the public $(N, H, pk_A)$, can learn $x_i$ with non-negligible advantage.
\end{proposition}

\begin{proof}
The ciphertext $c_i = H^{sk_i}(1+N)^{x_i} \bmod N^2$ multiplies the Paillier encoding of $x_i$ by the mask $H^{sk_i}$ with a secret, uniformly chosen $sk_i$. By the masking assumption the mask is pseudorandom in $\langle H\rangle$, and by DCR the additive layer is semantically secure; together these make $c_i$ a semantically secure encryption of $x_i$~\cite{shi2011, joye2013, lem14b}, so no efficient adversary distinguishes it from an encryption of any other ballot. The public values $(N, H, pk_A)$ do not reveal the factorization of $N$ and give no additional advantage.
\end{proof}

\begin{proposition}[Collector obliviousness]\label{prop:collector}
The collector $C$, seeing $\{aux_i\}_{i=1}^n$, learns nothing about any $x_i$.
\end{proposition}

\begin{proof}
Each $aux_i = pk_A^{sk_i} \bmod N^2$ is a function of $(pk_A, sk_i)$ only and does not depend on the ballot $x_i$ in any way. The collector's entire view is therefore independent of the votes, so it learns nothing about them. (The collector never sees a ciphertext $c_i$.)
\end{proof}

\begin{proposition}[Aggregator obliviousness]\label{prop:agg}
The aggregator $A$, seeing $\{c_i\}_{i=1}^n$ and the aggregate $aux$, learns only $S = \sum x_i$ and nothing about individual $x_i$.
\end{proposition}

\begin{proof}
$A$'s view is $(\{c_i\}, aux)$, where $aux = H^{sk_A \cdot \Sigma}$ is a single group element; $A$ never sees an individual $aux_i$. The protocol yields $R = S$. Learning an individual $x_i$ would require distinguishing $c_i$ from an encryption of another ballot, which contradicts Proposition~\ref{prop:privacy}; the aggregate $aux$ encodes only $\Sigma = \sum sk_i$ in the exponent and exposes no individual $sk_i$. This is the aggregator-oblivious property proven in~\cite{shi2011, joye2013, lem14b}.
\end{proof}

\subsection{Integrity properties}

The architecture provides input integrity through the chain and output integrity through determinism:

\begin{itemize}[leftmargin=*]
    \item \textbf{On-chain immutability.} Each $c_i$ is a finalized Solana transaction and cannot be altered afterward, which gives a permanent audit trail of the exact ciphertext set.
    \item \textbf{Cross-checkable inputs.} The ciphertexts the aggregator reads from the database can be checked against the immutable on-chain record. Any insertion, deletion, or modification of a ballot at the host diverges from the chain and is detectable by any auditor, so host-side tampering with the input set is tamper-evident rather than silent.
    \item \textbf{Deterministic, recomputable tally.} Given $\{c_i\}$ and $aux$, the recovered sum $R$ is a deterministic function of public inputs. Anyone holding the on-chain ciphertexts and $aux$ can recompute $R$ and verify the published result; a dishonest host cannot publish a different valid tally.
    \item \textbf{Stage enforcement.} The on-chain program admits votes only during the Voting stage and candidate changes only during Application.
    \item \textbf{Parameter validation.} Before each run the parameters are checked ($N$ odd, $\gcd(sk_A, N) = 1$, $sk_A$ non-trivial, $N$ of the expected bit length), which rules out the degenerate configurations of Assumptions~\ref{ass:modulus} and~\ref{ass:key}.
\end{itemize}

These checks make tampering \emph{detectable}. Making it \emph{impossible} for a malicious host requires hardware isolation, which we treat as future work (Section~\ref{sec:future}).

\subsection{Limitations}

\begin{enumerate}[leftmargin=*]
    \item \textbf{No range proofs.} Assumption~\ref{ass:binary} is not enforced cryptographically, so a malicious voter could submit an out-of-range value. Zero-knowledge range proofs would close this gap.
    \item \textbf{Trusted setup.} $TP$ knows the factorization of $N$. Distributed key generation would remove this trust.
    \item \textbf{Single aggregator.} $A$ holds $sk_A$ and could withhold a result. Threshold decryption across several aggregators would mitigate this.
    \item \textbf{Collusion.} The privacy analysis assumes the collector and aggregator do not collude. A party holding both an individual $aux_i$ and the matching $c_i$ is outside the analyzed model; the secure-enclave direction of Section~\ref{sec:enclave} also helps here by isolating each role.
    \item \textbf{No coercion resistance.} A voter can reveal $sk_i$ to prove a vote. Receipt-freeness mechanisms~\cite{kiayias2008civitas} would be needed.
    \item \textbf{Local validator.} The test validator is not Byzantine fault-tolerant; mainnet deployment is needed for production security.
\end{enumerate}

\section{Vote Packing}\label{sec:packing}

To tally a multi-candidate election from a single ciphertext per voter, votes are packed into one integer using fixed-width bit slots. For $k$ candidates with $b$ bits per slot,
\begin{equation}
x_i = \sum_{j=1}^{k} v_{i,j} \cdot 2^{b(k-j)} = v_{i,1} \cdot 2^{b(k-1)} + v_{i,2} \cdot 2^{b(k-2)} + \cdots + v_{i,k} \cdot 2^{0}.
\end{equation}
Candidate 1 takes the most significant slot, candidate $k$ the least. The homomorphic sum keeps the slots independent:
\begin{equation}
\sum_{i=1}^n x_i = \sum_{j=1}^{k} \left(\sum_{i=1}^n v_{i,j}\right) \cdot 2^{b(k-j)},
\end{equation}
and extraction is $\mathrm{count}_j = \lfloor R / 2^{b(k-j)} \rfloor \bmod 2^b$. Figure~\ref{fig:packing} shows the layout for the implementation's choice $k = 10$, $b = 25$.

\begin{figure}[h]
\centering
\includegraphics[width=\textwidth]{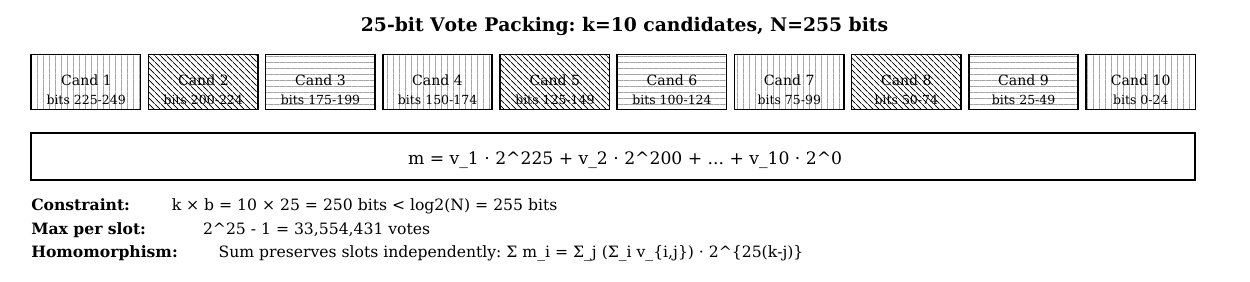}
\caption{Bit packing with the implementation's $b = 25$ and $k = 10$. Each slot is 25 bits, so $10 \times 25 = 250$ bits, below $\lfloor\log_2 N\rfloor = 254$. The maximum count per candidate is $2^{25}-1 = 33{,}554{,}431$.}
\label{fig:packing}
\end{figure}

\subsection{The slot width is a free parameter}\label{sec:packing-b}

The choice $b = 25$ is not fundamental. The only constraint packing imposes is Assumption~\ref{ass:packing}: the packed ballot must fit in $\mathbb{Z}_N$, i.e.\ $k \cdot b \le \lfloor \log_2 N \rfloor$. It is the \emph{product} $k \cdot b$ that the modulus bounds, not $b$ or $k$ individually. For a fixed $N$ one may trade slot width against candidate count: a 255-bit modulus accommodates $10 \times 25$, or $5 \times 50$, or $25 \times 10$, among others. A wider slot raises the per-candidate vote ceiling $2^b - 1$ but leaves room for fewer candidates; a narrower slot does the reverse. The implementation uses $b = 25$ because it balances ten candidates against a 33.5-million ceiling per candidate, which is comfortable for large elections, but any $(k, b)$ with $k \cdot b \le \lfloor \log_2 N \rfloor$ is valid.

\subsection{Why bit-packing matters}\label{sec:packing-adv}

The straightforward alternative is one ciphertext per candidate: voter $U_i$ would send $k$ separate Paillier ciphertexts $c_{i,1}, \ldots, c_{i,k}$. Packing replaces those $k$ ciphertexts with one and removes a factor of $k$ from almost every cost in the system (Figure~\ref{fig:packadv}).

\begin{figure}[h]
\centering
\includegraphics[width=\textwidth]{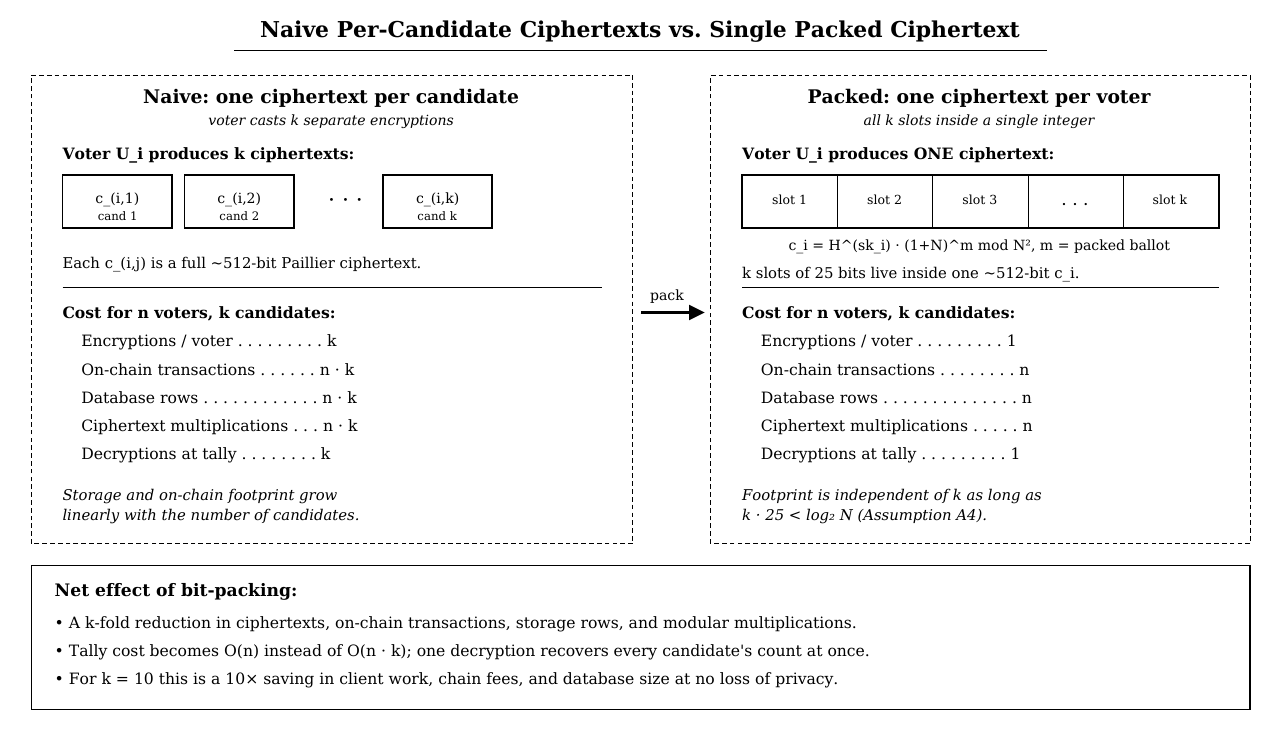}
\caption{Per-candidate ciphertexts versus a single packed ciphertext. Packing fixes the client, on-chain, storage, and tally cost at a level independent of the candidate count $k$, as long as $k \cdot b \le \lfloor\log_2 N\rfloor$.}
\label{fig:packadv}
\end{figure}

\begin{itemize}[leftmargin=*]
    \item \textbf{Client work.} The voter performs one Paillier encryption instead of $k$. On a phone or laptop doing exponentiation modulo the 510-bit $N^2$, this is the difference between one and $k$ of the most expensive operations in the flow.
    \item \textbf{On-chain footprint.} One ciphertext means one transaction and one stored account per voter instead of $k$. On Solana, where reserved account space costs rent, this divides the on-chain cost of an election by $k$.
    \item \textbf{Storage.} The database holds $n$ ciphertext rows instead of $n \cdot k$. At $\sim$405 bytes per voter row (Section~\ref{sec:evaluation}), a 10-candidate election stores about $10\times$ less.
    \item \textbf{Tally cost.} Aggregation multiplies $n$ ciphertexts rather than $n \cdot k$, so the running product is $O(n)$ instead of $O(n \cdot k)$. One decryption then recovers every candidate's count at once, instead of $k$ separate decryptions.
    \item \textbf{No privacy cost.} The packed ciphertext is still a single Paillier ciphertext and remains semantically secure under DCR (Proposition~\ref{prop:privacy}); packing changes only the plaintext encoding, not the hardness.
\end{itemize}

\subsection{Capacity analysis}

For the implementation's $b = 25$ and 255-bit $N$:
\begin{itemize}[leftmargin=*]
    \item \textbf{Maximum candidates:} $k_{\max} = \lfloor \lfloor\log_2 N\rfloor / b \rfloor = \lfloor 254 / 25 \rfloor = 10$.
    \item \textbf{Maximum votes per candidate:} $2^b - 1 = 33{,}554{,}431$.
    \item \textbf{Maximum total ballots:} up to $k_{\max} \times (2^b - 1) \approx 335.5$ million when votes are spread so that no candidate exceeds its slot. A worst-case landslide, in which a single candidate could receive every vote, is bounded by the per-slot ceiling $2^b - 1 = 33.5$ million.
\end{itemize}

For more candidates, or a higher per-candidate ceiling, a larger $N$ is needed (Table~\ref{tab:scaling}).

\begin{table}[h]
\centering
\caption{Modulus size required for a given candidate count, at $b=25$.}
\label{tab:scaling}
\begin{tabular}{rrrr}
\toprule
Candidates ($k$) & Bits needed ($k \times b$) & Min $N$ bits & Min prime size \\
\midrule
10 & 250 & 255 & 128-bit \\
20 & 500 & 505 & 253-bit \\
50 & 1{,}250 & 1{,}255 & 628-bit \\
100 & 2{,}500 & 2{,}505 & 1{,}253-bit \\
\bottomrule
\end{tabular}
\end{table}

\section{Algorithms}\label{sec:algorithms}

\begin{algorithm}[h]
\caption{Voter encryption (client-side, WebAssembly)}
\label{alg:encrypt}
\begin{algorithmic}[1]
\Require Vote vector $\mathbf{v} = [v_1, \ldots, v_k] \in \{0,1\}^k$; public params $(N, H)$; aggregator public key $pk_A$; slot width $b$
\Ensure Ciphertext $c$; auxiliary value $aux$

\Statex \textit{Assumptions:} \ref{ass:packing} ($k \cdot b \le \lfloor\log_2 N\rfloor$), \ref{ass:binary} ($v_j \in \{0,1\}$), \ref{ass:keyrand} ($sk$ reused as randomness $r$)

\State $sk \xleftarrow{\$} [2, N^2)$ \Comment{secret key, also the masking randomness}
\State $m \gets 0$
\For{$j = 1$ to $k$}
    \State $m \gets m + v_j \cdot 2^{b(k-j)}$ \Comment{pack votes in reverse slot order}
\EndFor
\State $c \gets \textsc{ModExp}(H,\, sk,\, N^2) \cdot \textsc{ModExp}(1{+}N,\, m,\, N^2) \bmod N^2$
\State $aux \gets \textsc{ModExp}(pk_A,\, sk,\, N^2)$
\State \Return $(c,\; aux)$
\end{algorithmic}
\end{algorithm}

\begin{algorithm}[h]
\caption{Vote aggregation (server-side, native libtommath)}
\label{alg:aggregate}
\begin{algorithmic}[1]
\Require Ciphertexts $\{c_1, \ldots, c_n\}$; collected auxiliary $aux$; params $(N, H, sk_A)$; candidate count $k$; slot width $b$
\Ensure Vote counts $[\mathrm{count}_1, \ldots, \mathrm{count}_k]$

\Statex \textit{Assumptions:} \ref{ass:modulus} ($N{=}pq$, odd primes), \ref{ass:key} ($\gcd(sk_A,N){=}1$), \ref{ass:sum} ($S{<}N$)

\State $N^2 \gets N \cdot N$
\State $sk_A^{-1} \gets \textsc{ModInv}(sk_A \bmod N,\; N)$ \Comment{exists by Assumption~\ref{ass:key}}

\Statex \Comment{\textbf{Step 1: product of ciphertexts (Theorem~\ref{thm:product})}}
\State $\Pi \gets 1$
\For{$i = 1$ to $n$}
    \State $\Pi \gets \Pi \cdot c_i \bmod N^2$ \Comment{\texttt{mp\_mul} + \texttt{mp\_mod}}
\EndFor

\Statex \Comment{\textbf{Step 2: exponentiate by $sk_A$}}
\State $P \gets \textsc{ModExp}(\Pi,\, sk_A,\, N^2)$ \Comment{\texttt{mp\_exptmod}}

\Statex \Comment{\textbf{Step 3: cancel mask (Theorem~\ref{thm:mask})}}
\State $aux^{-1} \gets \textsc{ModInv}(aux,\, N^2)$ \Comment{exists by Assumption~\ref{ass:modulus}}
\State $P' \gets P \cdot aux^{-1} \bmod N^2$

\Statex \Comment{\textbf{Step 4: $L$-function}}
\State $L \gets (P' - 1) \;\mathbf{div}\; N$ \Comment{exact integer division}

\Statex \Comment{\textbf{Step 5: recover sum (Theorem~\ref{thm:sum})}}
\State $R \gets L \cdot sk_A^{-1} \bmod N$ \Comment{now $R = S = \sum x_i$}

\Statex \Comment{\textbf{Step 6: unpack votes (Corollary~\ref{cor:votes})}}
\For{$j = 1$ to $k$}
    \State $\mathrm{count}_j \gets (R \gg b(k-j)) \;\&\; (2^{b}-1)$ \Comment{right-shift + bitmask}
\EndFor
\State \Return $[\mathrm{count}_1, \ldots, \mathrm{count}_k]$
\end{algorithmic}
\end{algorithm}

\section{Proof-of-Concept Implementation}\label{sec:implementation}

The architecture is implemented as a running, open-source proof-of-concept that validates the design.\footnote{Source code: \url{https://github.com/Akshit11318/narad}} The numbers reported below, and in the appendix, are taken from this deployment rather than chosen by hand.

\subsection{Implementation architecture}\label{sec:arch}

Each role of Section~\ref{sec:model} maps to a software component (Figure~\ref{fig:architecture}). Client-side cryptography runs in the browser; the collector and aggregator run in the backend, which calls a native C module for the heavy arithmetic; PostgreSQL and Solana provide storage and an audit trail.

\begin{figure}[h]
\centering
\includegraphics[width=\textwidth]{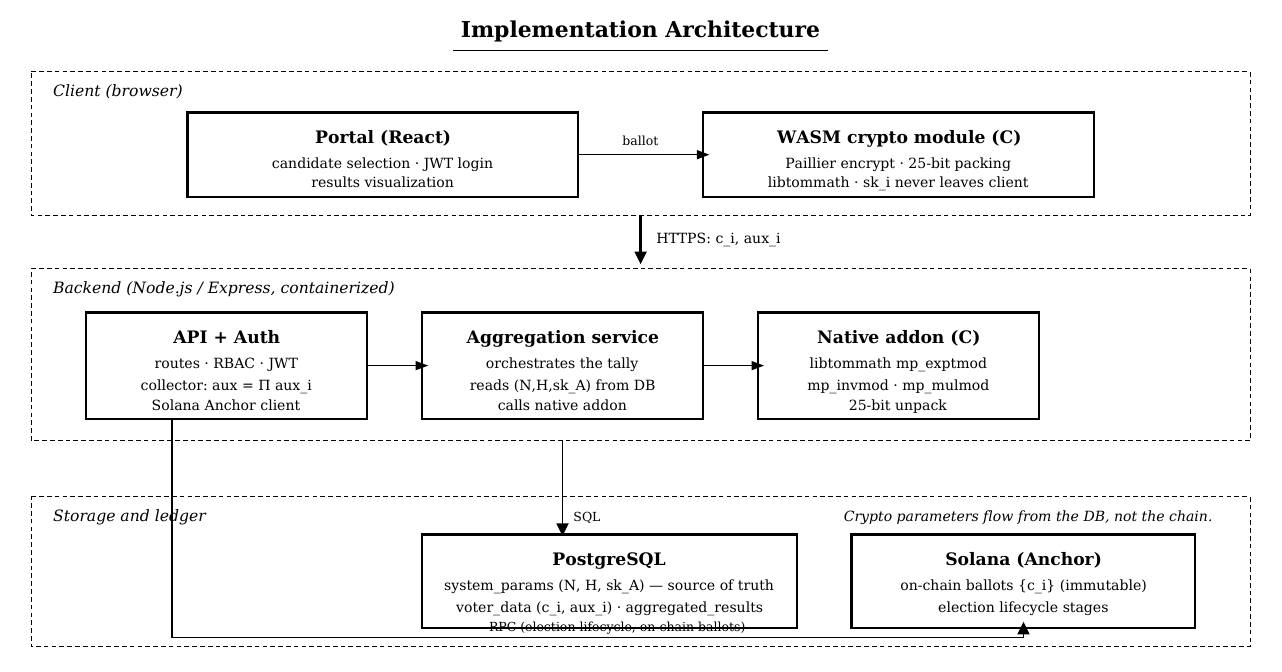}
\caption{Implementation architecture. Client-side encryption runs in the browser through WebAssembly. The collector and aggregator both run in the backend, which delegates the modular arithmetic to a native C addon. PostgreSQL is the single source of truth for the Paillier parameters; Solana provides an immutable ballot record.}
\label{fig:architecture}
\end{figure}

\paragraph{Voter (browser + WASM).} All cryptography runs client-side in a WebAssembly module compiled from C with Emscripten and linked against libtommath. The plaintext $x_i$ and the secret key $sk_i$ never leave the browser. The module exports \texttt{generate\_secret\_key} ($sk_i$), \texttt{pack\_votes} (the packed integer $m$), \texttt{encrypt\_vote} ($c_i = H^{sk_i}(1+N)^{m} \bmod N^2$), and \texttt{compute\_aux} ($aux_i = pk_A^{sk_i} \bmod N^2$).

\paragraph{Collector.} A backend service that multiplies the auxiliary values into $aux = \prod aux_i \bmod N^2$. It can be split into a separate party where stronger separation of duties is required.

\paragraph{Aggregator (native C / libtommath).} Runs the tally of Section~\ref{sec:protocol}. Every modular exponentiation and inversion happens in a native C addon linked against the same libtommath used in the WASM client, so the two sides agree bit-for-bit.

\paragraph{Blockchain (Solana).} An Anchor program records each $c_i$ as an append-only transaction and enforces the election lifecycle (application $\to$ voting $\to$ closed). It does not participate in aggregation; it provides immutability, not computation.

\paragraph{Database (PostgreSQL).} The \texttt{system\_params} table holds $(N, H, sk_A)$ and is the single source of truth for the Paillier parameters. The \texttt{voter\_data} table holds the per-voter $c_i$ and $aux_i$, and the \texttt{aggregated\_results} table holds the tally.

\paragraph{Design rationale.} Encryption runs in WebAssembly because libtommath's \texttt{mp\_exptmod} executes exponentiation modulo the 510-bit $N^2$ at near-native speed in the browser sandbox, several times faster than JavaScript \texttt{BigInt}. Aggregation runs in a native addon because the tally is $n$ modular multiplications plus one large exponentiation, where native C clearly outperforms \texttt{BigInt} at scale. Parameters live in PostgreSQL rather than on-chain so a setup mistake can be corrected without redeploying an immutable program; the chain still records every ballot for an audit trail. The blockchain is included because homomorphic aggregation alone gives no audit trail: recording each $c_i$ on-chain lets any observer confirm that the ciphertext set used in the tally matches what was submitted.

\subsection{Technology stack}

\begin{itemize}[leftmargin=*]
    \item WASM: C + Emscripten 4.0.6 + libtommath
    \item Backend: Node.js 20, Express, Prisma ORM
    \item Native addon: C + Node-API, cmake-js, libtommath
    \item Blockchain: Solana (Anchor 0.30.1), test validator
    \item Database: PostgreSQL 16
    \item Frontend: React 19, Tailwind CSS
    \item Deployment: Docker Compose (four containers: backend, postgres, solana, portal)
\end{itemize}

\subsection{Concrete parameters}

The running system uses the following 255-bit parameters, read directly from the \texttt{system\_params} table through the aggregator parameter endpoint. The same values drive the evaluation in Section~\ref{sec:evaluation} and the worked example in Appendix~\ref{sec:example}.

\begin{lstlisting}
N   = 400000000000000000000000211d1b86000000000000000001b131395144546b
H   = 2e0f69afca2410fda28718e5623a7a755531ae6dd30a286ec6737b8b2a6a7b61
skA = 2481c95a9eb72624f93de30ccd1118fbccfb637cd367ad167433a866751833b
\end{lstlisting}

\noindent Here $N$ is a 255-bit odd modulus (Assumption~\ref{ass:modulus}), and $sk_A$ is coprime to $N$ and distinct from $H$ (Assumption~\ref{ass:key}). These checks run at the start of every test.

\subsection{Data storage strategy}

Ciphertexts $c_i$ and auxiliary values $aux_i$ are stored both on Solana (for immutability) and in PostgreSQL (for aggregation). The Paillier parameters live only in PostgreSQL: \texttt{system\_params} is the single source of truth. PostgreSQL also handles voter authentication (bcrypt password hashes, JWT tokens); the blockchain handles the election lifecycle and ballot recording.

\subsection{Engineering issues resolved}

Bringing the protocol from paper to a working system surfaced several concrete bugs, each tied to an assumption or to a low-level detail:

\begin{enumerate}[leftmargin=*]
    \item \textbf{Even modulus (Assumption~\ref{ass:modulus}):} a hardcoded $N$ was even. Fixed by generating $N = pq$ with odd primes.
    \item \textbf{Degenerate key (Assumption~\ref{ass:key}):} $sk_A$ was set equal to $H$. Fixed with a separate random $sk_A$.
    \item \textbf{Parameter source mismatch:} aggregation read $N$ from a stale on-chain value. Fixed by reading from the database.
    \item \textbf{Packing mismatch:} the C path used 25-bit slots while a JavaScript fallback used 32-bit. Unified to 25-bit.
    \item \textbf{Missing $sk_A^{-1}$ step:} the JavaScript path omitted $L(P') \cdot sk_A^{-1}$. Added.
    \item \textbf{GCD pre-check bug:} a faulty check rejected a valid $sk_A$. Removed the pre-check in favor of a direct \texttt{mp\_invmod}.
    \item \textbf{Byte order:} a native binding emitted big-endian where the rest of the code expected little-endian. Added the reversal.
\end{enumerate}

The byte-order and parameter-source issues in particular show the gap between a clean protocol and a correct implementation: the mathematics was right long before the bytes lined up.

\section{Evaluation}\label{sec:evaluation}

\subsection{Methodology}

Measurements were taken on a 12th-generation Intel Core i7-12700H (20 hardware threads) with 16~GB of RAM, running the four-container Docker Compose stack. The test harness reads the live parameters $(N, H, sk_A)$ from the backend, encrypts ballots in Python with the same Paillier formulas as the WASM client, bulk-loads them into PostgreSQL with \texttt{COPY}, and then drives the real backend endpoints for the collector and aggregator steps. The aggregation timings are therefore the native libtommath tally, not a re-implementation. Each run checks the recovered per-candidate counts against the counts generated during encryption.

\subsection{Correctness and performance}

Table~\ref{tab:results} reports five runs at different scales. In every run the recovered counts matched the expected counts exactly. The encryption column reflects the Python harness standing in for the client; the aggregation column is the native tally and is the figure of interest for server cost.

\begin{table}[h]
\centering
\caption{Measured results on the parameters of Section~\ref{sec:implementation}. ``Encrypt'' is harness-side Paillier encryption plus bulk insert; ``Aux'' is the collector product; ``Aggregate'' is the native libtommath tally. All runs recovered exact counts.}
\label{tab:results}
\begin{tabular}{rrrrrrc}
\toprule
Voters & Cand. & Encrypt (s) & Aux (s) & Aggregate (s) & DB size & Correct \\
\midrule
100      & 10 & 0.07  & 0.007 & 0.018 & 39 KB   & yes \\
1{,}000  & 5  & 0.44  & 0.012 & 0.041 & 394 KB  & yes \\
5{,}000  & 5  & 2.02  & 0.039 & 0.104 & 1.9 MB  & yes \\
10{,}000 & 10 & 4.17  & 0.052 & 0.189 & 3.9 MB  & yes \\
50{,}000 & 10 & 21.77 & 0.199 & 0.815 & 19.4 MB & yes \\
\bottomrule
\end{tabular}
\end{table}

The native tally scales linearly with the number of ballots, from 18~ms at 100 voters to 0.82~s at 50{,}000, since the running product is $O(n)$ and the single decryption that follows is $O(1)$ in the voter count. Effective aggregation throughput rises with scale as the fixed setup cost amortizes, from about 5{,}400 ballots/s at 100 voters to about 61{,}000 ballots/s at 50{,}000. Per-voter storage is steady at roughly 405 bytes (two $\sim$512-bit hex values plus metadata), so a 50{,}000-ballot election occupies about 19~MB.

\subsection{Capacity}

With $b = 25$ and the 255-bit modulus, each candidate slot holds up to $2^{25}-1 = 33{,}554{,}431$ votes, and ten candidates give a total ceiling near 335 million ballots, bounded by Assumptions~\ref{ass:packing} and~\ref{ass:binary}.

\section{Scaling and Future Directions}\label{sec:future}

The same design accommodates much larger or smaller elections by turning a few independent dials, and it extends naturally to richer ballot formats and stronger integrity. Figure~\ref{fig:scaling} summarizes the levers.

\begin{figure}[h]
\centering
\includegraphics[width=\textwidth]{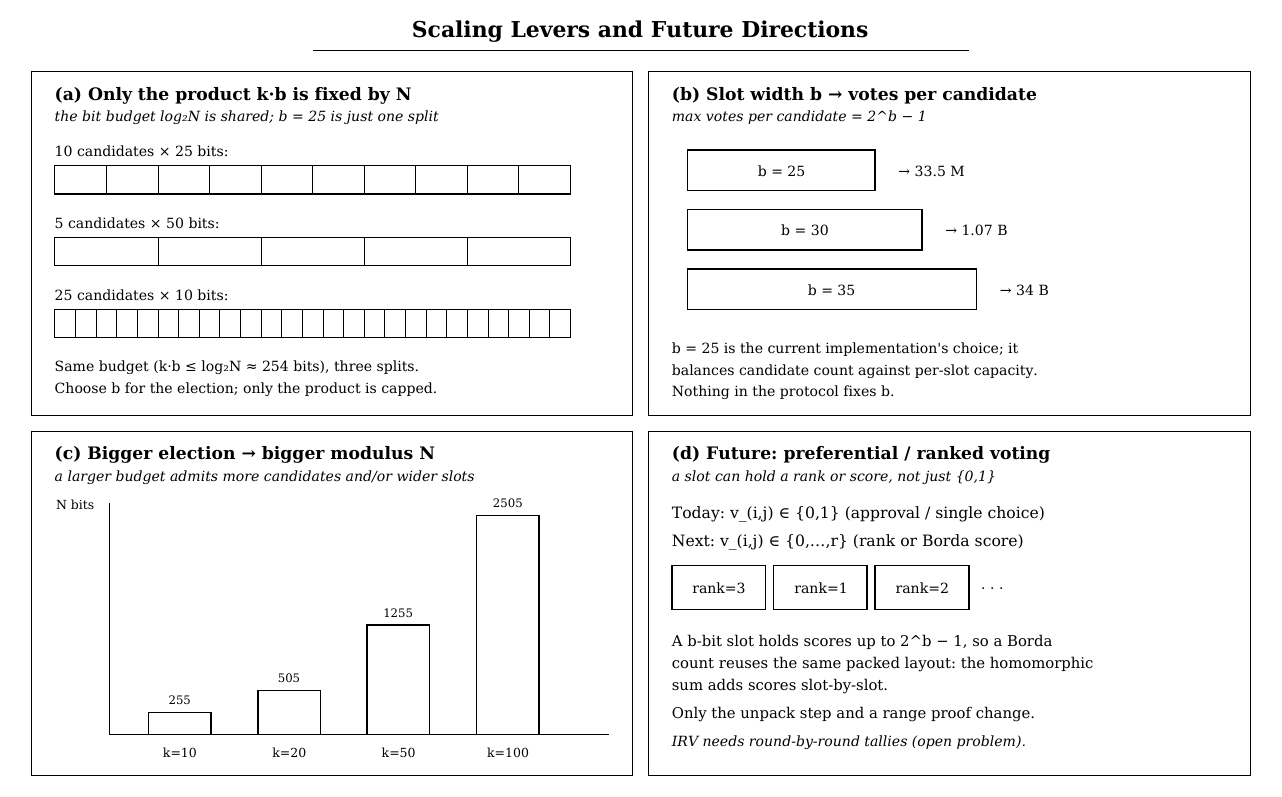}
\caption{Scaling levers and future directions: (a) for a fixed modulus only the product $k \cdot b$ is bounded, so the same bit budget can be split into many narrow slots or few wide ones; (b) wider slots admit more votes per candidate; (c) a larger modulus admits more candidates and/or wider slots; (d) preferential voting reuses the same packed layout with non-binary slot values.}
\label{fig:scaling}
\end{figure}

\subsection{More or fewer candidates}

The candidate count is bounded by $k \le \lfloor \lfloor\log_2 N\rfloor / b \rfloor$. With the current 255-bit modulus and $b = 25$ this is 10. Doubling to 20 candidates needs about a 505-bit modulus; 100 candidates need about 2{,}505 bits (Table~\ref{tab:scaling}). Raising $N$ costs more per modular operation but changes nothing else in the protocol or the proofs. A small election with two or three candidates can keep the 255-bit modulus and leave the unused high slots at zero, or move to a smaller modulus for faster arithmetic.

\subsection{More or fewer voters}

Two dials control voter capacity. The first is the slot width $b$: each slot holds up to $2^b - 1$ votes per candidate, so $b = 25$ gives 33.5 million, $b = 30$ gives about 1.07 billion, and $b = 35$ gives about 34 billion. Widening slots uses more of the bit budget, so the product $k \cdot b$ must still satisfy Assumption~\ref{ass:packing}; in practice one chooses $(k, b)$ to fit the target election. The second dial is the modulus size itself, which sets the budget $\lfloor\log_2 N\rfloor$ shared between $k$ and $b$. The tally remains $O(n)$ regardless of these choices.

\subsection{On-chain cost versus capacity}

On Solana, accounts pay rent for the space they reserve, and each ciphertext is an element of $\mathbb{Z}_{N^2}^*$, roughly $2\lceil|N|/8\rceil \approx 64$ bytes for a 255-bit $N$. The size of an election therefore trades directly against its on-chain cost: a larger rent budget reserves larger accounts and records more voters and candidates on-chain, with no change to the off-chain database or the homomorphic tally. An operator willing to pay more for stronger immutability can scale the on-chain portion up; one who needs only a lightweight audit trail can keep on-chain accounts small and lean on the database. The cryptographic protocol is identical in both cases.

\subsection{Preferential and ranked voting}

The packing scheme already admits non-binary slot values. Today $v_{i,j} \in \{0,1\}$ encodes approval or a single choice. Allowing $v_{i,j} \in \{0, \ldots, r\}$ encodes a rank or a Borda-style score, and the homomorphic sum still adds the values slot by slot, so a Borda count needs no new ciphertext layout: a $b$-bit slot holds scores up to $2^b - 1$. Only the unpack step and the validity (range) proof change. Fully ranked methods such as instant-runoff need round-by-round retallying of partial orders, which the additive scheme does not capture directly; supporting them efficiently is an open problem.

\subsection{Tamper-resistance through secure enclaves}\label{sec:enclave}

The integrity checks of Section~\ref{sec:security} make tampering \emph{detectable}: the on-chain record fixes the ciphertext set, and the tally is a deterministic, publicly recomputable function of it, so a host cannot silently alter inputs or forge a result. They do not, however, make tampering \emph{impossible}: a malicious host operator could still manipulate the in-memory computation of the collector or aggregator, or selectively process inputs, and would only be caught after the fact by an auditor who recomputes.

Running the collector inside a hardware secure enclave (for example Intel SGX, AMD SEV-SNP, or AWS Nitro Enclaves) would close this gap. The collector is a natural candidate: it touches every $aux_i$ and must faithfully compute $aux = \prod aux_i \bmod N^2$ without dropping or substituting values. Inside an enclave its code and memory are isolated from the host, including a root operator, and remote attestation lets any party verify the exact binary that is running before trusting its output. Extending the same treatment to the aggregator would additionally seal $sk_A$ and guarantee a faithful decryption. Combined with the on-chain audit trail, enclaves would move the system from tamper-evidence to tamper-resistance: no party, not even the machine owner, could read intermediate state or alter the votes, and anyone could check, via attestation, that the published tally came from the unmodified protocol.

Enclaves also address the collusion limitation of Section~\ref{sec:security}. Ballot privacy depends on the collector and aggregator not pooling their views: the collector's individual values $\{aux_i\}$ together with the aggregator's ciphertexts $\{c_i\}$ would expose the per-voter masking exponents and break privacy. In a single-operator deployment, where both roles run on the same infrastructure, only operational discipline keeps the two apart. Placing each role in its own attested enclave makes that separation cryptographic rather than procedural: each enclave ingests only its designated inputs (the collector the $aux_i$, the aggregator the $c_i$ and the aggregate $aux$), computes only its defined function, and emits only its defined output, with attestation proving it does nothing more. A host with root access can neither read an enclave's memory nor reroute the individual $aux_i$ into the aggregator, so collusion across the two roles becomes infeasible even when one party owns all the hardware. This is the strongest reason to favor enclaves here: they convert the non-collusion assumption from a trust requirement into a hardware-enforced guarantee.

\subsection{Threshold aggregation across multiple aggregators}\label{sec:threshold}

The current design uses a single aggregator holding $sk_A$. A malicious aggregator can withhold results (though it cannot forge them). Threshold decryption distributes $sk_A$ across $t$ parties using $(k, t)$-secret sharing, so no single party can decrypt alone.

In a Paillier threshold scheme~\cite{damgard2001}, the key $sk_A$ is split into shares $sk_{A,1}, \ldots, sk_{A,t}$ via Shamir secret sharing over $\mathbb{Z}_{\lambda}$ (where $\lambda = \mathrm{lcm}(p-1, q-1)$). Each aggregator $A_j$ computes a partial decryption $P_j = \Pi^{sk_{A,j}} \bmod N^2$ from the ciphertext product $\Pi$. The partial decryptions are then combined via Lagrange interpolation to recover $P = \Pi^{sk_A} \bmod N^2$, after which the rest of the protocol (mask cancellation, $L$-function, sum recovery) proceeds as before.

This has two benefits. \emph{First}, it removes the single-aggregator availability bottleneck: as long as $k$ of $t$ aggregators participate, the tally can proceed. \emph{Second}, it hardens the collector's position: with a single aggregator, the collector forwards $aux$ to one party who then knows both $\{c_i\}$ and $aux$. With threshold aggregation, each aggregator sees only $\Pi$ (the public ciphertext product) and its own partial decryption share; no single aggregator ever holds enough information to cancel the mask alone, so the collector's $aux$ is distributed across parties rather than handed to one. This makes it harder for any single host to link individual ciphertexts to auxiliary values, strengthening the non-collusion property even without enclaves.

The cost is a more complex setup (shared key generation) and a slightly taller tally (one Lagrange interpolation step). For elections requiring high integrity, the trade-off is favorable.

\subsection{Further work}

Beyond the directions above, the priorities are: integrating the zero-knowledge range proofs described in Section~\ref{sec:related} to enforce vote validity (Assumption~\ref{ass:binary}); implementing the threshold aggregation of Section~\ref{sec:threshold} to remove single-aggregator trust; distributed key generation to remove the trusted setup entirely; and a Solana mainnet deployment with formal verification of the on-chain program.

\section{Conclusion}\label{sec:conclusion}

We described an architecture for privacy-preserving blockchain voting based on Paillier homomorphic encryption. The central idea is that each voter reuses a self-generated secret key as both the masking randomness and the auxiliary exponent, which lets the aggregator cancel the random masks without any key coordination. We proved correctness under six explicit assumptions and analyzed privacy under DCR. A bit-packing scheme places a full multi-candidate ballot in one ciphertext and removes a factor of $k$ from client work, on-chain footprint, storage, and tally cost; its slot width is a free parameter, with only the product $k \cdot b$ bounded by the modulus. With $b = 25$ and a 255-bit modulus the design supports ten candidates and on the order of 335 million ballots, and the proof-of-concept tallies 50{,}000 ballots exactly in under a second using native libtommath arithmetic. The same design scales to larger elections by enlarging the modulus or the slot width, trades on-chain cost against capacity through Solana rent, extends to score-based preferential voting within the existing packing model, and can be made tamper-resistant by running the collector and aggregator inside attested secure enclaves.

\appendix
\section{Worked Example on the Real Parameters}\label{sec:example}

This appendix runs the full protocol for 3 voters and 2 candidates using the actual system parameters of Section~\ref{sec:implementation}. Every value below is computed, not invented; the script in \texttt{figures/} (\texttt{worked\_example.py}) reproduces it and asserts the final result. Large group elements are shown abbreviated as \texttt{head\dots tail} in hexadecimal; the small plaintexts and the recovered sum are exact.

\subsection{Setup}

The published parameters are the real $N$, $H$, and $sk_A$:
\begin{lstlisting}
N   = 400000000000000000000000211d1b86000000000000000001b131395144546b
H   = 2e0f69afca2410fda28718e5623a7a755531ae6dd30a286ec6737b8b2a6a7b61
skA = 2481c95a9eb72624f93de30ccd1118fbccfb637cd367ad167433a866751833b
\end{lstlisting}
\noindent The aggregator publishes $pk_A = H^{sk_A} \bmod N^2 = \texttt{90101d091d\dots d2da8b7c}$. The three voters fix the secret keys (random in the real system, fixed here for reproducibility):
\begin{lstlisting}
sk1 = 1f2e3d4c5b6a79887766554433221100ffeeddccbbaa99887766554433221101
sk2 = 0a1b2c3d4e5f60718293a4b5c6d7e8f90112233445566778899aabbccddeeff1
sk3 = 123456789abcdef0fedcba98765432100123456789abcdef0fedcba987654321
\end{lstlisting}

\subsection{Encryption}

With $b = 25$, the ballots and their packings are
\[U_1 \to \text{cand 1}: \mathbf{v}_1 = [1,0],\ m_1 = 2^{25} = 33{,}554{,}432; \quad U_2 \to \text{cand 2}: \mathbf{v}_2 = [0,1],\ m_2 = 1;\]
\[U_3 \to \text{cand 1}: \mathbf{v}_3 = [1,0],\ m_3 = 2^{25} = 33{,}554{,}432.\]
The expected tally is candidate 1 $= 2$, candidate 2 $= 1$. Each voter computes $c_i = H^{sk_i}\,(1+N)^{m_i} \bmod N^2$ and $aux_i = pk_A^{sk_i} \bmod N^2$:
\begin{lstlisting}
c1   = 5240f0f7e2...0c56b81c      aux1 = 7a3776619b...e818e3b9
c2   = 4196c34735...4270970f      aux2 = 5d8eb1c992...175f09c5
c3   = f7e600014f...a33897e7      aux3 = c7f87d7ef6...84da6db1
\end{lstlisting}

\subsection{Collection}

The collector multiplies the auxiliary values modulo $N^2$:
\[aux = aux_1 \cdot aux_2 \cdot aux_3 \bmod N^2 = \texttt{66153be6a0\dots 8860c768}.\]

\subsection{Aggregation}

The aggregator runs the five steps of Section~\ref{sec:protocol}:
\begin{align*}
\Pi  &= c_1 c_2 c_3 \bmod N^2 = \texttt{da1145ab6e\dots ab5271d6} \\
P    &= \Pi^{sk_A} \bmod N^2 = \texttt{452b3db4cb\dots fa66bac3} \\
P'   &= P \cdot aux^{-1} \bmod N^2 = \texttt{63bf297ccc\dots 0b1ae769} \\
L    &= (P' - 1)/N = (sk_A \cdot S) \bmod N = 1127918153\dots 070612024 \quad (\text{77 digits}) \\
R    &= L \cdot sk_A^{-1} \bmod N = 67{,}108{,}865
\end{align*}

\subsection{Result}

The recovered sum is $R = 67{,}108{,}865 = 2 \cdot 2^{25} + 1$, exactly the sum of the packed plaintexts $m_1 + m_2 + m_3$. Unpacking the two 25-bit slots gives
\[\mathrm{count}_1 = \lfloor R / 2^{25} \rfloor \bmod 2^{25} = 2, \qquad \mathrm{count}_2 = R \bmod 2^{25} = 1,\]
which matches the expected tally of two votes for candidate 1 and one for candidate 2. The script verifies $R = m_1 + m_2 + m_3$ and $[\mathrm{count}_1, \mathrm{count}_2] = [2, 1]$, both of which hold.

This is the same arithmetic the native aggregator performs at scale. In the 100-voter, 10-candidate run of Section~\ref{sec:evaluation}, the recovered counts were
\begin{lstlisting}
Alice 7  Bob 11  Carol 8  Dave 15  Eve 7
Frank 10  Grace 10  Heidi 7  Ivan 11  Judy 14   (total 100, all exact)
\end{lstlisting}



\begin{thebibliography}{99}

\bibitem{paillier1999public}
P.~Paillier, ``Public-key cryptosystems based on composite degree residuosity classes,'' in \emph{Proc. EUROCRYPT}, 1999, pp.~223--238.

\bibitem{adida2008helios}
B.~Adida, ``Helios: Web-based open-audit voting,'' in \emph{Proc. USENIX Security Symposium}, 2008, pp.~335--348.

\bibitem{kiayias2008civitas}
A.~Kiayias, M.~Korman, and D.~Walluck, ``Voting without self-enforcing protocols: A survey of e-voting schemes,'' in \emph{Proc. ACNS}, 2008.

\bibitem{shi2011}
E.~Shi, T.-H.~H. Chan, E.~Rieffel, R.~Chow, and D.~Song, ``Privacy-preserving aggregation of time-series data,'' in \emph{Proc. NDSS}, 2011.

\bibitem{joye2013}
M.~Joye and B.~Libert, ``A scalable scheme for privacy-preserving aggregation of time-series data,'' in \emph{Proc. Financial Cryptography and Data Security (FC)}, 2013, pp.~111--125.

\bibitem{lem14b}
I.~Leontiadis, K.~Elkhiyaoui, and R.~Molva, ``Private and dynamic time-series data aggregation with trust relaxation,'' in \emph{Proc. Cryptology and Network Security (CANS)}, 2014.

\bibitem{popets2023}
M.~Mansouri, M.~\"Onen, W.~Ben~Jaballah, and M.~Conti, ``SoK: Secure aggregation based on cryptographic schemes for federated learning,'' \emph{Proceedings on Privacy Enhancing Technologies (PoPETs)}, vol.~2023, no.~1, pp.~140--157, 2023.

\bibitem{turbo2021}
J.~So and B.~G\"uler, ``Turbo-Aggregate: Breaking the quadratic aggregation barrier in secure federated learning,'' \emph{IEEE Journal on Selected Areas in Information Theory}, vol.~2, no.~1, pp.~479--489, 2021.

\bibitem{libtommath}
``libtommath: A portable number theoretic multiple-precision integer library,'' [Online]. Available: \url{https://github.com/libtom/libtommath}

\bibitem{emscripten}
``Emscripten: LLVM to WebAssembly compiler,'' [Online]. Available: \url{https://emscripten.org}

\bibitem{anchor}
``Anchor: Solana Sealevel Framework,'' [Online]. Available: \url{https://github.com/coral-xyz/anchor}

\bibitem{solana}
``Solana: A fast, secure, and censorship-resistant blockchain,'' [Online]. Available: \url{https://solana.com}

\bibitem{docker}
``Docker: Containerized application platform,'' [Online]. Available: \url{https://docker.com}

\bibitem{pedersen1991}
T.~P. Pedersen, ``Non-interactive and information-theoretic secure verifiable secret sharing,'' in \emph{Proc. CRYPTO}, 1991, pp.~129--140.

\bibitem{schnorr1991}
C.~P. Schnorr, ``Efficient signature generation by smart cards,'' \emph{Journal of Cryptology}, vol.~4, no.~3, pp.~161--174, 1991.

\bibitem{bunz2018bulletproofs}
B.~B\"unz, S.~Bootle, D.~Boneh, A.~Poelstra, P.~Wuille, and G.~Maxwell, ``Bulletproofs: Short proofs for confidential transactions and more,'' in \emph{Proc. IEEE Symposium on Security and Privacy (S\&P)}, 2018, pp.~315--334.

\bibitem{camenisch1997}
J.~Camenisch and M.~Stadler, ``Proof systems for general statements about discrete logarithms,'' Tech. Report TR~260, Dept. of Computer Science, ETH Z\"urich, 1997.

\bibitem{damgard2001}
I.~Damg\aa rd and M.~Jurik, ``A generalisation, a simplification and some applications of Paillier's probabilistic public-key system,'' in \emph{Proc. PKC}, 2001, pp.~119--136.

\bibitem{sgaxe2020}
S.~van~Schaik, A.~Kwong, D.~Genkin, and Y.~Yarom, ``SGAxe: How SGX fails in practice,'' 2020, [Online]. Available: \url{https://sgaxeattack.com/}

\bibitem{lvi2020}
J.~Bulck, M.~Minkin, O.~Weisse, D.~Genkin, B.~Kasikci, F.~Piessens, M.~Silberstein, T.~Witchel, and Y.~Yarom, ``LVI: Hijacking transient execution through microarchitectural data injection,'' in \emph{Proc. IEEE Symposium on Security and Privacy (S\&P)}, 2020.

\end{thebibliography}
\end{document}